\documentclass[11pt,a4paper,pdftex]{article}

\usepackage[left=28mm,right=28mm,top=30mm,bottom=30mm]{geometry}

\usepackage[colorlinks,unicode]{hyperref}
\usepackage{amsmath}
\usepackage{amssymb}
\usepackage{amsthm}
\usepackage{stmaryrd}
\usepackage{mathtools}
\usepackage{mathdots}
\usepackage{braket}
\usepackage{enumerate}
\usepackage{cite}
\usepackage{appendix}
\usepackage{bm}

\newtheorem{thm}{Theorem}
\newtheorem{prop}{Proposition}

\def\n{\mathfrak{n}}

\def\bZ{\mathbb{Z}}

\def\g{\mathfrak{g}}
\def\h{\mathfrak{h}}

\def\Z22{\mathbb{Z}_{2}^2}

\def\nn{\nonumber}

\def\ol{\overline}
\title{A classification of lowest weight irreducible modules over $\Z22$-graded extension of $osp(1|2)$}
\author{K. Amakawa, N. Aizawa\footnote{Corresponding author: aizawa@p.s.osakafu-u.ac.jp}
	\\[10pt]
Department of Physical Science, Osaka Prefecture University, \\
Nakamozu Campus, Sakai, Osaka 599-8531, Japan}

\date{March 11, 2021}

\begin{document}
	\maketitle
	\thispagestyle{empty}
	
	\vfill
	\begin{abstract}
		We investigate representations of the $\mathbb{Z}_2^2$-graded extension of $osp(1|2)$ which is the spectrum generating algebra of the recently introduced $\mathbb{Z}_2^2$-graded version of superconformal mechanics. The main result is a classification of irreducible lowest weight modules of the $\mathbb{Z}_2^2$-graded extension of $osp(1|2)$. 
		This is done via introduction of  Verma modules and its maximal invariant submodule generated by singular vectors.  Explicit formula of all singular vectors are also presented. 
	\end{abstract}

	\clearpage
	\setcounter{page}{1}

 \section{Introduction}
 
Conformal mechanics introduced by de Alfaro, Fubini,  Furlan \cite{DFF} and its supersymmetric extensions 
have been studied extensively in theoretical and mathematical physics (see, e.g., \cite{BelKri,SCM1,SCM2} and references therein).  
Recently, it was shown that many models of superconformal mechanics discussed in the literature can be 
generalized to $ \Z22 := \mathbb{Z}_2 \times \mathbb{Z}_2$ setting \cite{AAD} 
where the superconformal algebras characterizing the models were replaced with their $ \Z22$-graded extension. 
As a result, the $\Z22$-graded version of superconformal mechanics has a richer structure than the standard one. 

The $\Z22$-graded superconformal algebra is an example of the $\Z22$-graded Lie superalgebras introduced by 
Rittenberg and Wyler \cite{rw1,rw2}.   
In fact, the algebra of Rittenberg-Wyler is a special case of the more general algebras considered  by Ree.
The Ree's algebra is referred to as colour Lie algebra or $\epsilon$-Lie algebra \cite{Ree} (see also \cite{sch1,sch2}) and 
the colour Lie algebras are defined on a vector space which is graded by an Abelian group $A$. 
In the case $ A = \mathbb{Z}_2,$ one may recover the  standard Lie superalgebras, 
thus the color Lie algebras are a natural generalization of Lie superalgebra. 
In the more general case of 
$ A = \mathbb{Z}_2^n := \mathbb{Z}_2 \times \mathbb{Z}_2 \times \cdots \times \mathbb{Z}_2 $ ($n$ copies of $\mathbb{Z}_2$), 
which is the one considered in \cite{rw1,rw2}, the colour Lie algebras are realized  in terms of commutator and anticommutator 
(a more precise definition of $ A = \Z22$ case is found in Appendix A), 
so that they may generate a continuous symmetry as standard Lie (super)algebras do. 
The continuous symmetries, generated by the algebra of higher gradation than $\mathbb{Z}_2$,  
were considered in several physical problems \cite{LuRi,vas,jyw,zhe,LR,Toro1,Toro2,tol2,tol}.
Furthermore, the recent  works recognize such higher graded  symmetries in 
more fundamental level, namely, 
non-relativistic wave equations for spinors \cite{AKTT1,AKTT2}, 
extensions of various models of supersymmetric theories \cite{AAD,Bruce,BruDup,AAd2,AKTcl,AKTqu,brusigma,Topp}. 
This implies that the higher graded symmetries should be studied more seriously. 

In order to explore symmetry, we need representation theories of the algebra generating the symmetry. 
However, the present knowledge of representations of colour Lie algebras is far from completion 
since only the limited number of works has been done so far  \cite{ChSiVO,SigSil,MohSal,NAJap,NAJS,NAPIJS,StoVDJ,Meyer}. 
Motivated by this situation, in the present work, we study the lowest weight Verma modules over a $\Z22$-graded extension of 
$ osp(1|2) $ ($\Z22$-$osp(1|2)$ in short) 
and present a classification of its irreducible lowest weight modules. 
This is done by the method similar to the standard Lie theory. 
As a result we classify $\Z22$-graded modules with a lowest weight of degree $(0,0)$. 
This approach is legitimated  since we show that the following two facts are also true for  $\Z22$-$osp(1|2)$: 
(i) the $\Z22$-graded Verma module has a unique maximal invariant submodule 
and the quotient of the Verma module by the maximal invariant module is irreducible, 
(ii) the maximal invariant submodule is generated by singular vectors. 
Thus the crucial step to the classification is the explicit construction of the singular vectors. 
This is carried out by an elementary method and we present the explicit formula of all singular vectors in the given Verma module.

The reason for choosing the $\Z22$-$osp(1|2)$ is two fold. 
First of all, the superalgebra $ osp(1|2)$ is the most fundamental Lie superalgebra. 
It plays the same role as $ sl(2) $ in the classical Lie algebras. 
Thus one may think that $\Z22$-$osp(1|2)$ plays the similar role in $\Z22$-graded Lie superalgebras. 
If so,  studying its representations in detail will be a basis for further development of the representation theories of 
more general $\Z22$-graded Lie superalgebras.  

The second reason is  physical applications. 
The $\Z22$-$osp(1|2)$ gives the simplest example of the $\Z22$-graded version of superconformal mechanics. 
In \cite{AAD}, the standard procedure of the conformal mechanics is employed to compute energy eigenvalues of the model. 
The energy spectrum is discrete and, reflecting the $\Z22$-grading, 
the ground state is doubly degenerate. 
The energy spectrum corresponds to an infinite dimensional $\Z22$-$osp(1|2)$-module with the degenerate lowest weight vectors.    
It is known that the standard ($\mathbb{Z}_2$-graded) Lie superalgebra $osp(1|2)$ has finite dimensional lowest weight 
modules of odd dimension in addition to the infinite dimensional ones.   
Then concerning the results in \cite{AAD} one may pose the following  question:  
does the $\Z22$-$osp(1|2)$ also has finite dimensional modules or 
additional modules which have no counterparts in the standard $osp(1|2)$ ? 
In the present work, we positively answer the question by showing  the existence of finite dimensional and 
another infinite dimensional irreducible modules over the $\Z22$-$osp(1|2)$ where the lowest weight has no degeneracy. 
Therefore, structure of the lowest weight modules over the  $\Z22$-$osp(1|2)$ is richer than the standard $osp(1|2).$ 
Since the both finite and infinite dimensional representations of $ osp(1|2)$ are relevant to various physical problems, 
one may expect more physical applications of $\Z22$-$osp(1|2)$. 

Before discussing the classification of the irreducible modules, 
we mention that various mathematical studies on color Lie algebras besides the representation theory  
have been undertaken since their introduction. 
Algebraic aspects of them are investigated from various directions 
\cite{sch3,SchZha,Sil,PionSil,CART,NAcl,Meyer2,IsStvdJ,NAPSIJSinf}. 
One of the recent hot topics is the geometry of $\mathbb{Z}_2^n$-graded manifolds 
which is an extension of the geometry of supermanifolds 
\cite{CGP1,CGP2,CGP3,CGP4,Pon,CoKPo,BruIbar,BruDupq,BruPon,
BruIbarPonc,BruGraRiemann,BruGrabow,CoKwPon,BruIbarPon}. 
This $\mathbb{Z}_2^n$-graded supergeometry is also relevant to consider the physical implications of 
$\mathbb{Z}_2^n$-graded Lie superalgebras.

We plan the paper as follows. 
The algebra $\Z22$-$osp(1|2)$ considered in this work is presented  in \S \ref{SEC:Alg}. 
It is pointed out that the Cartan subalgebra consists of the elements of different $\Z22$-degree 
and that, in order to investigate representations, one may employ the method similar to the strange Lie superalgebra $Q(n).$ 
We thus need to begin with a classification of irreducible modules of Cartan subalgebra as they are not necessary one dimensional.  
Verma modules over the $\Z22$-$osp(1|2)$ are introduced in \S \ref{SEC:Verma}. 
Then we state our main theorem (Theorem \ref{thm:ILWM}) which is a classification of irreducible lowest weight modules. 
To prove the theorem, we first show  that the Verma modules have properties similar to the case of standard Lie (super)algebras.  
Especially, quotient of a Verma module by its maximal invariant submodule produces an irreducible module 
and the maximal invariant submodule is constructed on singular vectors. 
Thus we make a search of all singular vectors in \S \ref{SEC:Sing} and present an explicit formula of all of them (Theorem \ref{thm:sv}). 
With the knowledge on the singular vectors, we make a list of all irreducible modules in \S \ref{SEC:IrMod}.

Throughout this article, all vector spaces are considered over the ground field $\mathbb{C}. $  

 %
 \section{$\Z22$-$osp(1|2)$ algebra} \label{SEC:Alg}

\subsection{Definitions and triangular decomposition}
 
Here we give the definition of  $\Z22$-$osp(1|2) $ which was introduced first in \cite{rw2}.   
From now on we denote $\Z22$-$osp(1|2)$ simply by $\g $ and follow the convention of \cite{AAD}. 
By definition (see Appendix A for detail), 
as a vector space, $\g$ consists of four subspaces each of which is labelled by an element of $\Z22: $ 
\[
  \g = \g_{(0,0)}\oplus\g_{(0,1)}\oplus\g_{(1,0)}\oplus\g_{(1,1)}. 
\]
Each even subspace ($\g_{(0,0)}$ and $\g_{(1,1)}$) is of three dimension and each odd subspace is of two dimension 
so that $ \dim \g = 10$.  
The basis of $\g$
\begin{equation}
   R, \ L_{\pm} \in \g_{(0,0)}, 
   \quad
   a_{\pm} \in \g_{(0,1)}, \quad \tilde{a}_{\pm} \in \g_{(1,0)},
   \quad \tilde{R},\ \tilde{L}_{\pm} \in \g_{(1,1)}
\end{equation}
is subject to the following non-vanishing relations:
\begin{alignat}{4}
 [R,L_{\pm}] &=\pm 2L_{\pm}, &\quad 
 [R,\tilde{L}_{\pm}] &=\pm 2\tilde{L}_{\pm}, &\quad 
 [R,a_{\pm}]& =\pm a_{\pm}, &\quad 
 [R,\tilde{a}_{\pm}]&=\pm \tilde{a}_{\pm}, 
 \nn\\
 [\tilde{R}, L_{\pm}]&=\pm 2\tilde{L}_{\pm}, &\quad 
 [\tilde{R},\tilde{L}_{\pm}] &=\pm 2L_{\pm} ,&\quad 
 \{ \tilde{R}, a_{\pm} \} &= \tilde{a}_{\pm}, &\quad 
 \{ \tilde{R}, \tilde{a}_{\pm} \}& = a_{\pm},
 \nn\\ 
 [L_{+},L_{-}]&=-R, &\quad 
 [L_{\pm}, \tilde{L}_{\mp}] &=\mp\tilde{R},  &\quad  
 [\tilde{L}_{+}, \tilde{L}_{-}] &=-R, 
 \nn\\
 [ L_{\pm},\tilde{a}_{\mp} ]&=\pm\tilde{a}_{\pm}, &\quad 
 [ L_{\pm}, a_{\mp} ]&=\mp a_{\pm}, &\quad 
 \{ \tilde{L}_{\pm}, a_{\mp} \} &=- \tilde{a}_{\pm}, &\quad 
 \{ \tilde{L}_{\pm}, \tilde{a}_{\mp} \}& = a_{\pm}, 
 \nn\\
 \{ a_{+}, a_{-} \} &=2R, &\quad 
 [ a_{\pm},\tilde{a}_{\mp} ]&=\pm 2\tilde{R}, &\quad 
 \{ \tilde{a}_{-}, \tilde{a}_{+} \}&=2R,
 \nn\\
 [ a_{\pm},\tilde{a}_{\pm} ]&=\mp 4\tilde{L}_{\pm}, &\quad 
 \{ a_{\pm}, a_{\pm} \} &=4L_{\pm}, &\quad 
 \{ \tilde{a}_{\pm}, \tilde{a}_{\pm} \}&=-4L_{\pm}. \label{Z22osp12}
\end{alignat} 

A Cartan subalgebra of $\g$ is defined in a way analogous to Lie superalgebras. 
Namely, it is defined as the maximal nilpotent subalgebra of $\g$ coinciding with its own normalizer. 
Then $\h := \mathrm{lin.\ span}\{  \ R, \ \tilde{R} \ \} $ is a Cartan subalgebra of $\g$ 
and $\h$ is a subspace of $ \g_{(0,0)} \oplus \g_{(1,1)}. $ 
One may see from \eqref{Z22osp12} that $ L_{\pm}, \tilde{L}_{\pm} $ and all the elements of $ \g_{(0,1)} \oplus \g_{(1,0)}$ are not eigenvectors of 
$ \mathrm{ad} \tilde{R}$ which is defined by $ \mathrm{ad} \tilde{R}(X) := \llbracket \tilde{R}, X \rrbracket $ for $ X \in \g$ 
where $ \llbracket \ , \ \rrbracket $ is the general Lie bracket (see Appendix A).  
Therefore, it is not possible to introduce a triangular decomposition of $\g$ according to the adjoint action of $\h$ and 
this is an obstacle to define the lowest weight modules. 

However, this is the situation similar to the case of the strange Lie superalgebra $Q(n)$ where the Cartan subalgebra 
consists of both even and odd elements. 
To overcome the difficulty, the root decomposition of $Q(n)$ is defined with respect to only the even part of its Cartan subalgebra.  
This procedure allows us to define and study the lowest (highest) weight modules of $Q(n)$ 
with a minimal modification from the cases of another simple Lie superalgebra \cite{muss,dict,pen}.

Thus we employ the method of $Q(n)$ to study our algebra $\g.$ 
Namely, we introduce a triangular decomposition of $\g$ with respect to only the degree $(0,0)$ part of $\h,$ 
more precisely, with respect to the eigenvalue of $\mathrm{ad} R.$ 
The $\Z22$-degree and the eigenvalue of $\mathrm{ad} R$ of the elements of $\g$ are summarized below:
\begin{equation}
   \begin{array}{c|cccc}
         & (0,0) & (0,1) & (1,0) & (1,1) \\\hline
      +2 & L_{+} & & & \tilde{L}_{+}
      \\
      +1 & & a_{+} & \tilde{a}_{+}
      \\
      0  & R & & & \tilde{R}
      \\
      -1 & & a_{-} & \tilde{a}_{-}
      \\
      -2 &  L_{-} & & & \tilde{L}_{-}   
   \end{array} 
   \label{adj}
\end{equation}
We now have the triangular decomposition of $\g:$ 
\begin{equation}
  \g = \n^- \oplus \h \oplus \n^+,\quad 
  \n^{\pm} = \mathrm{lin. \ span} \{ \ L_{\pm}, \ a_{\pm},\ \tilde{a}_{\pm},\ \tilde{L}_{\pm}  \ \}.
\end{equation}

We want to define Verma modules over $\g$ as modules induced from 
irreducible modules of the Borel subalgebra $\mathfrak{b} := \h \oplus \n^-.$ 
It is easy to see that an irreducible $\h$-module $\nu$ is endowed with the structure of an irreducible $\mathfrak{b}$-modules 
by setting $\n^- \nu = 0 $ \cite{pen}. 
However, the  price to pay for the employed method is the fact that irreducible $\h$-modules are not necessary one dimensional 
as $\tilde{R}$ is not necessary diagonal in the modules. 
Therefore, we need to begin with a classification of irreducible $\h$-modules.

\subsection{Irreducible modules of $\h$}

 The following theorem is the classification of all  finite dimensional irreducible $\Z22$-graded  $\h$-modules.  
It shows that a finite dimensional irreducible $\h$-module is at most two dimensional. 

%
%
\begin{thm}\label{thm:car}
	Finite dimensional irreducible $\Z22$-graded module of $\h$ is either one dimensional or two dimensional.
	\begin{enumerate}
	 \renewcommand{\labelenumi}{\normalfont (\roman{enumi})}
	  \item one dimensional module $\nu(r) := \mathrm{lin.\ span}\{ \ \ket{0} \ \}$ with  weight $r$ is given by
	  \begin{equation}
	    R\ket{0} = r\ket{0},\quad \tilde{R}\ket{0} = 0,
	  \end{equation}
	  \item two dimensional module $ \nu(r,\lambda) := \mathrm{lin. \ span}\{ \ \ket{0}, \ \ket{1} \ \}$ 
	  with weight $r$ and one additional parameter $\lambda \neq 0$ is given by 
	  \begin{equation}
	      R\ket{0} = r\ket{0},\quad \tilde{R}\ket{0} = \ket{1},\quad \tilde{R}\ket{1} = \lambda\ket{0}
	  \end{equation}
	\end{enumerate}
	where the $\Z22$-degree of $\ket{0}$ and $\lambda$ are $(0,0).$ 
\end{thm}
\begin{proof}
Let $\ket{0}$ be a vector satisfying $R\ket{0} = r\ket{0}.$ 
Then the vector space spanned by the vectors
\begin{equation}
  \ket{0}, \quad \ket{1} := \tilde{R} \ket{0}, \quad \ket{2} := \tilde{R}^2 \ket{0}, \quad \ket{3}:= \tilde{R}^3 \ket{0}, \quad \dots
\end{equation}
is endowed with the structure of an infinite dimensional $\h$-module. 
To have a finite dimensional $\h$-module, this sequence of vectors has to terminate at some positive integer. 
Noting that the $\Z22$-degree of the vector $\ket{k}$ is $(0,0)$ for even $k$ and  $(1,1)$ for odd $k,$ 
one may impose the following condition to get an $n$ dimensional module:
\begin{equation}
   \tilde{R}\ket{n-1} 
   =
   \begin{cases}
      \displaystyle \sum_{j=0}^{{(n-3)}/{2}} c_{2j+1}\ket{2j+1}, & n:odd
      \\[2em]
      \displaystyle \sum_{j=0}^{{(n-2)}/{2}} c_{2j}\ket{2j}, & n:even
   \end{cases}
   \label{coditionh-mod}
\end{equation} 
It is understood that $ \tilde{R} \ket{0} = 0 $ for $n=1.$ 

The case $n=1,$ which is identical to $\nu(r),$ is the trivial one dimensional module so that obviously irreducible. 
If $n > 1$ is an odd integer, then it is immediate to see that the subspace spanned by $\ket{1}, \ket{2}, \cdots, \ket{n-1} $ 
is an invariant submodule of $\h.$ Thus there exist no irreducible modules for odd $n > 1.$ 

Next we suppose that $n$ is even. For any non-vanishing vector $ \ket{w_{00}} $ of degree $(0,0)$, 
$ \ket{w_{11}} := \tilde{R} \ket{w_{00}} $ is also non-vanishing but has degree $(1,1)$. 
If the relation 
\begin{equation}
   \tilde{R}\ket{w_{11}} = t \ket{w_{00}} \label{tRw11}
\end{equation}
holds for some constant $t$, 
then there exists an invariant submodule. 
It is easy to see that the invariant submodule for $t=0$ is one dimensional spanned by $ \ket{w_{11}}$ alone and for $ t \neq 0$ is 
two dimensional spanned by $ \ket{w_{00}} $ and $ \ket{w_{11}}.$ 
Therefore, for $t=0$ the quotient module by the one dimensional invariant submodule is odd dimensional so that it is reducible unless the quotient module is isomorphic to $\nu(r).$ 
On the other hand, for a non-vanishing $t$ the quotient modules is also even dimensional so that one may repeat the same argument 
and at the end we reach the two dimensional module $ \nu(r,\lambda)$ whose irreducibility is obvious.

We thus need to show that the relation \eqref{tRw11} holds true for $ n > 2, $ i.e., existence of such non-vanishing $t.$ 
The vectors $ \ket{w_{00}} $ may be written as a linear combination of $ \ket{k}$:
\begin{equation}
  \ket{w_{00}} = \sum_{j=0}^{(n-2)/2} \lambda_{2j} \ket{2j}.
\end{equation}
Then we have
\begin{equation}
  \ket{w_{11}} = \sum_{j=0}^{(n-2)/2} \lambda_{2j} \ket{2j+1}. 
\end{equation}
The condition \eqref{tRw11} together with \eqref{coditionh-mod} gives the following relations of the expansion coefficients: 
\begin{align}
 \lambda_{n-2} \, c_{0} &= t \, \lambda_{0},
 \\
 \lambda_{2j-2} + \lambda_{n-2}\, c_{2j} &= t\, \lambda_{2j}, \quad  j=1, 2, \cdots, \dfrac{n-2}{2} 
\end{align}
These relations allows us to express $ \lambda_{2j} $ in terms of $ \lambda_{n-2}:$
\begin{align}
    \lambda_{0} &= t^{-1} c_0\, \lambda_{n-2}, \nn \\
	\lambda_{2} &= t^{-2} (c_{0} + c_{2}t) \lambda_{n-2}, \nn \\
	\lambda_{4} &= t^{-3} (c_{0} + c_{2}t + c_{4}t^2) \lambda_{n-2}, \nn \\
	& \vdots \nn\\
	\lambda_{n-4} &= t^{-(n-2)/2} (c_{0} + c_{2}t + \cdots + c_{n-6}t^{(n-6)/2} + c_{n-4}t^{(n-4)/2}) \lambda_{n-2}, 
	\nn \\
	 \lambda_{n-4} &= (-c_{n-2} + t) \lambda_{n-2}. \label{lm4}
\end{align}
If $ \lambda_{n-2} = 0, $ then all $ \lambda_{2j}$ vanishes so that $ \ket{w_{00}}= 0.$ 
So we are able to assume that $ \lambda_{n-2} \neq 0.$ 
Then elimination of $\lambda_{n-4}$ from the last two relations in \eqref{lm4} gives an algebraic equation for $t:$
\begin{equation}
  t^{n/2} - c_{n-2}t^{(n-2)/2}- c_{n-4}t^{(n-4)/2}  -  c_{n-6}t^{(n-6)/2} - \cdots - c_{2} t - c_{0}= 0.
\end{equation}
This equation has at least one non-zero solution unless all $c_{2j}$ vanishes. 
Vanishing $ c_{2j}$ implies, due to the condition \eqref{coditionh-mod}, 
that $ \ket{n-1}$ spans the one dimensional invariant submodule so that the problem is reduced 
to the case of odd $n.$ 
This establishes  the existence of non-vanishing $t$ and completes the proof. 
\end{proof}

\section{Verma modules over $\g$ and their reducibility} \label{SEC:Verma}
\setcounter{equation}{0}

The irreducible $\h$-modules $ \nu(r)$ and $\nu(r,\lambda)$ (Theorem \ref{thm:car}) are endowed with the structure of 
an irreducible $\mathfrak{b}$-module by setting $ \n^- \ket{0} = 0, $ 
since $ \llbracket \n^-, \tilde{R} \rrbracket \subseteq \n^- $ as is seen from \eqref{Z22osp12}. 
Thus one may define Verma modules over $\g$ in a usual way, 
i.e., a module induced from the irreducible $\mathfrak{b}$-module. 
We have two types of Verma modules for $\g:$
\begin{align}
  M(r) &:= U(\g)\otimes_{U(\mathfrak{b})}\nu(r)  = U(\n^+) \otimes \nu(r),
  \\
  M(r,\lambda) &:= U(\g)\otimes_{U(\mathfrak{b})}\nu(r,\lambda) = U(\n^+) \otimes \nu(r,\lambda)
\end{align}
where $ U(\g) $ denotes the enveloping algebra of $\g.$

Our aim is the detailed study of reducibility of these Verma modules. 
Here we present the main result.
%
\begin{thm}\label{thm:ILWM}
		All irreducible $\Z22$-graded lowest weight modules over $\g$ are listed as follows:
		\begin{enumerate}
		 	 \renewcommand{\labelenumi}{\normalfont (\roman{enumi})}
		 \item the Verma module $ M(r)$ for $r+2M \neq 0, \qquad \dim M(r) = \infty $
		 \item the quotient module $ M(r)/W $ for $r + 2M =0, \quad \dim (M(r)/W) = (2M+1)^2$
		 \item the Verma module $ M(r,\lambda) $ for $ (r+2M)^2 \neq \lambda,\quad \dim M(r,\lambda) = \infty$
		 \item the quotient module $ M(r,\lambda)/W $ for $ (r+2M)^2 = \lambda, \quad \dim(M(r,\lambda)/W) = \infty $
		\end{enumerate}
		where $M $ is a non-negative integer and $W := U(\n^+) \omega $ is the maximal invariant submodule in the Verma module. 
$ \omega := lin. span\{  \ \ket{\chi_{01}}, \ \ket{\chi_{10}}\ \}$ is the two dimensional subspace of $ M(r) $ or $M(r,\lambda)$  spanned by the vectors given in Theorem \ref{thm:sv}. 
\end{thm}

The rest of this paper is dedicated to the proof of Theorem \ref{thm:ILWM}. 
A crucial point of the proof is the observation that the procedure same as simple Lie (super)algebras is also applicable to 
this case. Namely, the fact that taking a quotient by the  maximal invariant submodule produces an irreducible module 
and the maximal invariant submodule is generated by singular vectors. 
To see this we begin with basic properties of the Verma modules. 

For simplicity, we indicate $M(r)$ and $M(r,\lambda)$ by $\check{M}. $ 
The following two are immediate consequences of the defining relations \eqref{Z22osp12} (see also \eqref{adj}).
\begin{enumerate}
	\item  $\check{M}$ has weight space decomposition
	\begin{align}
	\check{M} = \bigoplus_{k\in\bZ_{\geq 0}} \check{M}_{k}, \quad \check{M}_{k} = \{ \ket{v} \in \check{M} \ |\ R\ket{v} = (r+k) \ket{v} \}
	\end{align}
	We refer to $k$  as \textit{level} of the weight space $\check{M}_{k}$.
	\item Each weight space is a direct sum of subspaces of degree either (0,0) and (1,1) or (0,1) and (1,0):
	\begin{align}
	\check{M}_{k} = \check{M}_{k,(0,0)} \oplus \check{M}_{k,(1,1)} \quad \text{or} \quad  
	\check{M}_{k,(0,1)} \oplus \check{M}_{k,(1,0)} \label{WSdegree}
	\end{align}
\end{enumerate}

The following is the key  proposition of the classification of the lowest weight irreducible modules.
%
%
\begin{prop}\label{prop:mis}
\leavevmode \par
  \begin{enumerate}
     \renewcommand{\labelenumi}{\normalfont (\roman{enumi})}  
     \item $\check{M}$ has a unique maximal submodule $W$.
     \item  $L := \check{M}/W$ is irreducible.
     \item For a positive integer $p$, let $\displaystyle W = W_{p} \oplus W_{p+1} \oplus \cdots$ be the weight space decomposition, then
       \begin{equation}
         a_{-}\ket{\chi} = \tilde{a}_{-} \ket{\chi} = 0, \quad \forall \ket{\chi} \in W_{p} \label{sv}
       \end{equation}
  \end{enumerate}
\end{prop} 
We remark that $ \tilde{L}_{-}\ket{\chi}=0$ follows immediately from \eqref{sv}. 
We refer to the vector $\ket{\chi}$ satisfying the relation \eqref{sv} as a \textit{singular vector}. 
Although one may prove Proposition \ref{prop:mis} in the same way as the case of simple Lie superalgebras 
(see, for instance \cite{muss}), 
we present the proof below as the proposition is crucial for the study of reducibility.
\begin{proof}
\noindent
(i)  Let $W$ be a submodule of $\check{M}$, then $W$ is a direct sum of the weight space $ W = \bigoplus_{k\in\bZ_{\geq 0}} W_{k}$. 
If $W_{0} \neq 0$, then $W_{0} = \nu(r)$ for $M(r)$ and $W_{0} = \nu(r,\lambda)$ for $M(r,\lambda)$. 
This is obvious for $M(r)$ as $\dim W_{0} = \dim \nu(r) = 1.$ 
For $ M(r,\lambda),$ it is immediate from $\tilde{R}W_{0} \subseteq W_{0}.$ 
Namely, if $\ket{0} \in W_{0}$, then $\ket{1} = \tilde{R}\ket{0}$ must be an element of $ W_{0}$ and vice versa. 
It follows from $W_{0} = \nu(r)$ or $ \nu(r,\lambda)$ that $W = \check{M}$ 
because $\check{M}$ is induced from $\nu(r)$ or $ \nu(r,\lambda). $
Hence if $W$ is a proper submodule, 
then $ W \subseteq \check{M}^+ := \bigoplus_{k=1} \check{M}_{k}$. 
Therefore, the sum of all proper submodule is contained in $\check{M}^+$, so this sum is the unique maximal submodule. 

\medskip
\noindent
(ii) Let $W$ be the maximal submodule of $\check{M}$ and $L=\check{M}/W$. 
Suppose that there exists a submodule $N \subset L$. 
Then for each element $\ket{u}\in N$, there exists an element $\ket{v}\in \check{M}$ 
such that $\ket{v}=\ket{u}+\ket{w}$ where $\ket{w}\in W$. 
Note that , by definition, $\ket{u}\notin W$. 
For all $X\in \g$, $X\ket{v}=X\ket{u}+X\ket{w}$. 
This means the existence of a submodule of $\check{M}$ larger than $W$ as $X\ket{u}\in N, X\ket{w}\in W$. 
This is a contradiction so that $N=0$. 

\medskip
\noindent 
(iii) If there exists $\ket{\chi}\in W_{p}$ such that $\tilde{a}_{-}\ket{\chi}\neq 0$ or $a_{-}\ket{\chi} \neq 0$, 
then  $W$ is not maximal  as $\tilde{a}_{-}\ket{\chi}, a_{-}\ket{\chi} \in \check{M}_{p-1}$.
\end{proof}

Due to Proposition \ref{prop:mis}, the classification problem is reduced to the search of all singular vectors 
of $M(r) $ and $M(r,\lambda).$  
In the next section, we perform the search and derive the explicit formula of all singular vectors.

\section{Singular vectors of the Verma modules} \label{SEC:Sing}
\setcounter{equation}{0}

The aim of this section is explicit construction of all singular vectors of $M(r)$ and $M(r,\lambda).$ 
To this end, we take the following basis of $M(r)$
\begin{equation}
  \ket{\alpha, k, m} := \tilde{a}_{+}^\alpha\, a_{+}^k\, \tilde{L}_{+}^m \ket{0} \label{BasisMr}
\end{equation}
and of $ M(r,\lambda)$
\begin{equation}
	\ket{\alpha, k, m; \beta} := \tilde{a}^\alpha\, a_{+}^k\, \tilde{L}_{+}^m \ket{\beta} \label{BasisMrl}
\end{equation}
where $ \alpha, \beta $ take a value of $0$ or $1$ and $k, m \in \bZ_{\geq 0}.$ 
Note that there is no need of higher powers of $\tilde{a}_+ $ and $L_+$ because of the relation  
$\tilde{a}_{+}^2 = -a_{+}^2 = -2 L_+$.
The action of $\g $ on these basis, which will be used in the search of singular vectors, 
is obtained by straightforward computation.  
{
\setlength{\leftmargini}{13pt} 
\begin{enumerate}
  \item action of $\h$ on $M(r):$
  \begin{align}
    R\ket{\alpha,k,m} &= (r+\alpha+k+2m)\ket{\alpha,k,m}, \label{RonMr}\\
    \tilde{R} \ket{0,k,m} &= (-1)^k m\ket{0,k+2,m-1} + (-1)^k 2(k-\bar{k})\ket{0,k-2,m+1} + \bar{k}\ket{1,k-1,m},\\
    \tilde{R} \ket{1,k,m} &= (-1)^{k+1} m\ket{1,k+2,m-1} + (-1)^{k+1} 2(k-\bar{k})\ket{1,k-2,m+1} \nn\\
        &+ (1+\bar{k})\ket{0,k+1,m}
  \end{align}
  where $ \bar{k} = k \mod 2.$
  \item action of $\n^+$ on $M(r):$
  \begin{align}
    \tilde{a}_{+}\ket{\alpha,k,m} &= (-1)^\alpha\ket{\overline{\alpha+1},k+2\alpha,m},\\
    a_{+}\ket{\alpha,k,m}&=\ket{\alpha,k+1,m} + (-1)^{k+1}  4\alpha \ket{0,k,m+1},\\
    \tilde{L}_{+}\ket{\alpha,k,m}&=(-1)^{\alpha+k}\ket{\alpha,k,m+1}.
  \end{align}
  \item  action of $\n^-$ on $M(r):$
  \begin{align}
    \tilde{a}_{-}\ket{0,k,m} &= -(k-\bar{k}) \ket{1,k-2,m} + (-1)^k  m\ket{0,k+1,m-1},\\
    \tilde{a}_{-}\ket{1,k,m} &= (2r+k+\bar{k}+4m )\ket{0,k,m} + (-1)^{k+1} m\ket{1,k+1,m-1},\\
    a_{-}\ket{0,k,m} &= (k+(2r-1)\bar{k})\ket{0,k-1,m} + (-1)^{k+1}  m\ket{1,k,m-1},\\
    a_{-}\ket{1,k,m} &= (k+(2r-3)\bar{k})\ket{1,k-1,m} +(-)^{k+1} m \ket{0,k+2,m-1} \nn\\
       &+ (-1)^{k+1}  4(k-\bar{k})\ket{0,k-2,m+1}, \\
    \tilde{L}_{-}\ket{0,k,m} &= (-1)^{k}  m(r+k+m-1)\ket{0,k,m-1} \nn \\
      &+ (-1)^k  (k-\bar{k})(k-\bar{k}-2)\ket{0,k-4,m+1} 
       + \bar{k}(k-1)\ket{1,k-3,m}, \\
   \tilde{L}_{-}\ket{1,k,m} &= (-1)^{k+1}  m(r+k+m)\ket{1,k,m-1} \nn \\
      & + (-1)^{k+1}(k-\bar{k})(k-\bar{k}-2)\ket{1,k-4,m+1} \nn\\
      & + (2(r-1) \bar{k} + (\bar{k}+1)k) \ket{0,k-1,m}.
   \end{align}
   \item action of $\h$ on $M(r,\lambda):$
   \begin{align}
      R\ket{\alpha,k,m;\beta} &= (r+\alpha+k+2m)\ket{\alpha,k,m;\beta}, \label{RonMrl}\\
      \tilde{R} \ket{0,k,m;\beta} &= (-1)^k m\ket{0,k+2,m-1;\beta} + (-1)^k 2(k-\bar{k})\ket{0,k-2,m+1;\beta} \nn\\
            & + \bar{k}\ket{1,k-1,m;\beta} + (-1)^k  \lambda^\beta\ket{0,k,m; \overline{\beta+1}}, \\
      \tilde{R} \ket{1,k,m;\beta} &= (-1)^{k+1} m\ket{1,k+2,m-1;\beta} + (-1)^{k+1}  2(k-\bar{k})\ket{1,k-2,m+1;\beta}\nn\\
            & + (1+\bar{k})\ket{0,k+1,m;\beta} + (-1)^{k+1}  \lambda^{\beta} \ket{1,k,m;\overline{\beta+1}}.
   \end{align}
   \item action of $\n^+$ on $M(r,\lambda):$
   \begin{align}
       \tilde{a}_{+}\ket{\alpha,k,m;\beta} &= (-1)^\alpha\ket{\overline{\alpha+1},k+ 2\alpha,m;\beta},\\
       a_{+}\ket{\alpha,k,m;\beta}&=\ket{\alpha,k+1,m} + (-1)^{k+1}  4\alpha\ket{0,k,m+1;\beta},\\\medskip
       \tilde{L}_{+}\ket{\alpha,k,m;\beta}&=(-1)^{\alpha+k}\ket{\alpha,k,m+1;\beta}.
   \end{align}
   \item action of $\n^-$ on $M(r,\lambda):$
   \begin{align}
       \tilde{a}_{-}\ket{0,k,m;\beta} &= -(k-\bar{k}) \ket{1,k-2,m;\beta} + (-1)^k  m\ket{0,k+1,m-1;\beta} \nn\\
           & - 2\bar{k}\lambda^\beta\ket{0,k-1,m;\overline{\beta+1}},\\
       \tilde{a}_{-}\ket{1,k,m;\beta} &= (2r+k+\bar{k}+4m )\ket{0,k,m;\beta} + (-1)^{k+1} m\ket{1,k+1,m-1;\beta} \nn\\
           & + 2\bar{k}\lambda^\beta\ket{1,k-1,m;\overline{\beta+1}}\\
      a_{-}\ket{0,k,m;\beta} &= (k+(2r-1)\bar{k})\ket{0,k-1,m;\beta} + (-1)^{k+1}  m\ket{1,k,m-1;\beta}, \\
      a_{-}\ket{1,k,m;\beta} &= (k+(2r-3)\bar{k})\ket{1,k-1,m;\beta} +(-1)^{k+1}m \ket{0,k+2,m-1;\beta} \nn\\
           &+ (-1)^{k+1}  4(k-\bar{k})\ket{0,k-2,m+1;\beta} + (-1)^{k+1}\lambda^\beta\ket{0,k,m;\overline{\beta+1}},\\
      \tilde{L}_{-}\ket{0,k,m;\beta} &= (-1)^{k}  m(r+k+m-1)\ket{0,k,m-1;\beta} \nn \\
           &+ (-1)^k  (k-\bar{k})(k-\bar{k}-2)\ket{0,k-4,m+1;\beta} \nn\\
           & + \bar{k}(k-1)\ket{1,k-3,m;\beta} + (-1)^k  (k-\bar{k})\lambda^\beta\ket{0,k-2,m;\overline{\beta+1}},\\
      \tilde{L}_{-}\ket{1,k,m;\beta} &= (-1)^{k+1}  m(r+k+m)\ket{1,k,m-1;\beta} \nn \\
           & + (-1)^{k+1}(k-\bar{k})(k-\bar{k}-2)\ket{1,k-4,m+1;\beta} \nn\\
           &+ (2(r-1) \bar{k} + (\bar{k}+1)k)\ket{0,k-1,m;\beta} \nn \\
           &+ (-1)^{k+1}  (k-\bar{k})\lambda^\beta\ket{1,k-2,m;\overline{\beta+1}}.
\end{align}
\end{enumerate}
}

%
%
\begin{thm}\label{thm:sv}
	The following is a complete list of singular vectors of the Verma modules $M(r)$ and $M(r,\lambda)$ over $\g$:
	\begin{enumerate}
	  \renewcommand{\labelenumi}{\normalfont (\roman{enumi})}
	   \item $M(r):$ For $r+2M=0$, there exist singular vectors of $\Z22$-degree  (0,1) and (1,0) at level $2M+1$ which are given by
		\begin{align}
		\ket{\chi_{01}} &= \sum_{j=0}^{\lfloor M/2 \rfloor} 2^{2j} \dbinom{M}{2j} \ket{0, 2(M-2j)+1, 2j} \nn\\
		&- \sum_{j=0}^{\lfloor (M-1)/2 \rfloor} 2^{2j+1} \dbinom{M}{2j+1} \ket{1, 2(M-2j-1), 2j+1}, \label{101}\\
		\ket{\chi_{10}} &= -\sum_{j=0}^{\lfloor (M-1)/2 \rfloor} 2^{2j+1} \dbinom{M}{2j+1} \ket{0, 2(M-2j-1)+1, 2j+1} \nn \\
		&+ \sum_{j=0}^{\lfloor M/2 \rfloor} 2^{2j} \dbinom{M}{2j} \ket{1, 2(M-2j), 2j} \label{110}
		\end{align} 
		and a degree (1,1) singular vector at level $2(2M+1)$ which is given by
       \begin{equation}
       	\ket{\chi_{11}} = \sum_{j=0}^{M} (-4)^j \dbinom{M}{j} ( -2 \ket{0,4(M-j),2j+1} + \ket{1,4(M-j)+1,2j} ) \label{111}
       \end{equation} 
       where $M \in \bZ_{\geq 0}.$ These singular vectors satisfy the relations:
       \begin{equation}
          \tilde{R}\ket{\chi_{01}} = (2M+1)\ket{\chi_{10}},\quad \tilde{R}\ket{\chi_{10}} = (2M+1)\ket{\chi_{01}},
           \quad
           \tilde{R}\ket{\chi_{11}} = 0. \label{rel0110}
       \end{equation}           
      \item $ M(r,\lambda):$ For $(r+2M)^2 = \lambda$ there exist degree (0,1) and (1,0) singular vectors 
      at level $2M+1\ (M \in \bZ_{\geq 0})$ which are given by
		\begin{align}
		\ket{\chi_{01}} &= \sum_{j=0}^{M} (-2)^j \dbinom{M}{j} 
		\big( (r+2M)^{\overline{j+1}}  \ket{0,2(M-j)+1,j;\overline{j}} 
		\nn \\
		& \hspace{5cm}+ (r+2M)^{\ol{j}}\ket{1,2(M-j),j;\overline{j+1}} \big), \label{201}\\
		\ket{\chi_{10}} &= \sum_{j=0}^{M} (-2)^j \dbinom{M}{j} 
		\big( (r+2M)^{\overline{j}}  \ket{0,2(M-j)+1,j;\overline{j+1}} 
		\nn \\
		& \hspace{5cm}+ (r+2M)^{\ol{j+1}} \ket{1,2(M-j),j;\overline{j}} \big)  \label{210}
		\end{align}
		where $\overline{j} = j \pmod 2$.  
		These singular vectors satisfy the relations
		\begin{equation}
         \tilde{R}\ket{\chi_{01}} = (-r+1)\ket{\chi_{10}},\quad \tilde{R}\ket{\chi_{10}} = (-r+1)\ket{\chi_{01}}. \label{relevenweightspace}		
       \end{equation}				
	\end{enumerate}
\end{thm}
\begin{proof}
We first note that the relations in \eqref{rel0110} and \eqref{relevenweightspace} are verified by direct computation of the action of $\tilde{R}$ on the explicit formulae of singular vectors. 
Thus we focus on the derivation of the explicit formulae. 

To perform a complete search of singular vectors we employ an elementary method: 
Since the singular vector $\ket{\chi}$ belongs to a weight space $\check{M}_N,$ 
$\ket{\chi}$ is a linear combination of the basis of $\check{M}_N. $ 
We determine the combination in such a way that the condition \eqref{sv} is satisfied. 
The $\Z22$-degree of the vectors in $\check{M}_N$ depends on the parity of $N$ (see \eqref{WSdegree}). 
Thus we need to treat even and odd $N$ separately. 
Before starting the search for singular vectors, 
observe that the level of the weight space $\check{M}_N$ is $N=\alpha+k+2m$ which is common for $M(r)$ and 
$M(r,\lambda)$ as seen from \eqref{RonMr} and \eqref{RonMrl}.

\medskip
\noindent
(i) $M(r): $ the degree of the basis \eqref{BasisMr} is $(\alpha+m, k+m)$.
First, we study the weight space of odd level.

\medskip
\noindent
(i-a) Consider a weight space of odd level $2M+1$. 
Then each vector has $\Z22$-degree either $(0,1)$ or $(1,0).$ 
The degree $(0,1)$ and $(1,0)$ basis are given as follows:
\[
\begin{array}{lll}
(0,1) \ : &\ \ket{0,2(M-2j)+1,2j}, &\quad 0 \leq j \leq \left\lfloor \frac{M}{2} \right\rfloor\\[3pt]
&\ \ket{1,2(M-2j-1),2j+1},& \quad 0 \leq j \leq \left\lfloor \frac{M-1}{2} \right\rfloor\\[3pt]
(1,0)\ : &\ \ket{0,2(M-2j-1)+1,2j+1},& \quad 0 \leq j \leq \left\lfloor \frac{M-1}{2} \right\rfloor\\[3pt]
&\ \ket{1,2(M-2j),2j},& \quad 0 \leq j \leq \left\lfloor \frac{M}{2} \right\rfloor
\end{array}
\]
The degree $(0,1)$ and $(1,0)$ singular vectors, denoted by $\ket{\chi_{01}}, \ket{\chi_{10}}$ respectively, 
have the form of
\begin{align*}
\ket{\chi_{01}} &= \sum\limits_{j=0}^{\left\lfloor \frac{M}{2} \right\rfloor} \mu_{j}\ket{0,2(M-2j)+1,2j} + \sum\limits_{j=0}^{\left\lfloor \frac{M-1}{2} \right\rfloor} \nu_{j}\ket{1,2(M-2j-1)+1,2j+1} ,\\
\ket{\chi_{10}} &= \sum\limits_{j=0}^{\left\lfloor \frac{M-1}{2} \right\rfloor} \alpha_{j}\ket{0,2(M-2j-1)+1,2j+1} + \sum\limits_{j=0}^{\left\lfloor \frac{M}{2} \right\rfloor} \beta_{j}\ket{1,2(M-2j),2j}
\end{align*}
and we determine the coefficients by imposing the condition \eqref{sv}. 
For clarity, we indicate explicitly the  coefficients to be determined:
\[
    \alpha_j,\ \nu_j, \quad  j = 0,1,\dots, 
\left\{
\begin{array}{ll}
K-1, & M=2K\\
K, & M=2K+1,
\end{array} \right.
\qquad
\beta_j,\ \mu_j, \quad j = 0, 1, \dots, K
\]

The condition $\tilde{a}_{-}\ket{\chi_{01}}=0$ gives the relations:
\begin{align}
2(M-2j)\mu_{j} + (2j+1)\nu_{j} = 0&, 
\qquad j = 0,1,\dots, 
\left\{
\begin{array}{ll}
K-1, & M=2K\\
K, & M=2K+1
\end{array}    
\right.  \label{101-1}\\
 (j+1)\mu_{j+1} - (r+M+2j+1)\nu_{j} = 0 &, \qquad j = 0,1,\dots, K-1 \label{101-2}
\end{align}
and the additional relation for $M=2K+1$
\begin{align}
(r+2M)\nu_{K} = 0. \label{101-3}
\end{align}
The condition $a_{-}\ket{\chi_{01}}=0$ gives the relations:
\begin{align}
 2(r+M) \mu_{0} - \nu_{0} = 0 &\label{101-4}\\
 (j+1)\mu_{j+1} + (M-2j-1)\nu_{j} = 0 &,\quad j=0,1,\dots, K-1 \label{101-5}\\
 2(r+M-2j)\mu_{j} - (2j+1)\nu_{j} -8(M-2j+1)\nu_{j-1} = 0 &, \label{101-6}\\
&\hspace{-2cm} j = 1, 2, \dots, 
\left\{
\begin{array}{ll}
K-1, & M=2K\\
K, & M=2K+1
\end{array}    
\right.\nn
\end{align}
and the additional relation for $M=2K$
\begin{align}
 r\mu_{K} - 4\nu_{K-1} = 0. \label{101-7}
\end{align}

Eliminating $\mu_j$ from  \eqref{101-2} and \eqref{101-5} we have
\begin{equation}
   (r+2M) \nu_j = 0, \quad 0 \leq j \leq K-1. \label{r2Mrel}
\end{equation}
If $(r+2M) \neq 0,$ then one may see that  $\mu_j = \nu_j = 0$ for all $j$ which means that $\ket{\chi_{01}}=0$.
For $(r+2M) = 0$, the recurrence relations are easily solved to give
\[
	\mu_j = 2^{2j} 
	   \begin{pmatrix}
	      M \\ 2j
	   \end{pmatrix} \mu_0, \qquad
    \nu_j = - 2^{2j+1} 
     \begin{pmatrix}
       M \\ 2j+1 
     \end{pmatrix} \mu_{0}.
\]
%
Thus the existence of $\ket{\chi_{01}}$ given by \eqref{101} has been shown. 

  We now turn to $\ket{\chi_{10}}.$
The condition $\tilde{a}_{-}\ket{\chi_{10}}=0$ gives the relations:
\begin{align}
  (M-2j-1)\alpha_{j} + (j+1)\beta_{j+1} = 0,& \quad j = 0,1,\dots,K-1 \label{110-4}\\
  -(2j+1)\alpha_{j} + 2(r+M+2j)\beta_{j} = 0,& \quad 
 j = 0,1,\dots, 
 \left\{
 \begin{array}{ll}
 K-1, & M=2K\\
 K, & M=2K+1
 \end{array}    
 \right. \label{110-5}
\end{align}
and the additional relation for $M=2K$
\begin{equation}
   (r+2M)\beta_{K} = 0. \label{110-6}
\end{equation}
The condition $a_{-}\ket{\chi_{10}}=0$ gives the relations:
\begin{align}
& (2j+1)\alpha_{j} + 2(M-2j)\beta_{j} = 0, 
\quad j = 0, 1, \dots, 
\left\{
\begin{array}{ll}
K-1, & M=2K\\
K, & M=2K+1
\end{array}    
\right. \label{110-1}\\
& (r+M-2j-1)\alpha_{j} - 4(M-2j)\beta_{j} - (j+1)\beta_{j+1} = 0, \quad j=0,1,\dots, K-1 \label{110-2} 
\end{align}
and the additional relation for $M=2K+1$
\begin{align}
r\alpha_{K} - 4\beta_{K} = 0. \label{110-3}
\end{align}

Comparing the above relations for $\ket{\chi_{10}}$ with those for $\ket{\chi_{01}}$ 
one may recognize that coefficients are related as $\alpha_j = \nu_j, \ \beta_j = \mu_j.$
This proves the existence of $\ket{\chi_{10}} $ and its explicit formula \eqref{110}. 

\medskip\noindent
(i-b) We now turn to the weight space of even level $2M.$ 
In this case each vector has $\Z22$-degree either $(0,0)$ or $(1.1)$. 
The degree $(0,0)$ and $(1,1)$ basis are given as follows:
\[
\begin{array}{lll}
(0,0)\ : &\ \ket{0, 2(M-2j), 2j}, & \quad 0 \leq j \leq \left\lfloor \frac{M}{2} \right\rfloor\\[3pt]
&\ \ket{1,2(M-2j-1)-1,2j+1},& \quad 0 \leq j \leq \left\lfloor \frac{M}{2} \right\rfloor - 1\\[3pt]
(1,1)\ : &\ \ket{0,2(M-2j-1),2j+1},& \quad 0 \leq j \leq \left\lfloor \frac{M-1}{2} \right\rfloor\\[3pt]
&\ \ket{1,2(M-2j)-1,2j},& \quad 0 \leq j \leq \left\lfloor \frac{M-1}{2} \right\rfloor
\end{array}
\]
The degree $(0,0)$ and $(1,1)$ singular vectors, denoted by $\ket{\chi_{00}}, \ket{\chi_{11}}$ respectively, 
have the form of
\begin{align*}
\ket{\chi_{00}} &= \sum\limits_{j=0}^{\left\lfloor \frac{M}{2} \right\rfloor} \rho_{j}\ket{0,2(M-2j),2j} + \sum\limits_{j=0}^{\left\lfloor \frac{M}{2} \right\rfloor -1} \sigma_{j}\ket{1,2(M-2j-1)-1,2j+1} ,\\
\ket{\chi_{11}} &= \sum\limits_{j=0}^{\left\lfloor \frac{M-1}{2} \right\rfloor} \gamma_{j}\ket{0,2(M-2j-1),2j+1} + \sum\limits_{j=0}^{\left\lfloor \frac{M-1}{2} \right\rfloor} \delta_{j}\ket{1,2(M-2j)-1,2j}.
\end{align*}
We impose the condition \eqref{sv} to determine the coefficients. 
More explicitly, we determine the following coefficients:
\begin{align*}
  \gamma_j, \ \delta_j, &\quad 
  j= 0, 1, \dots \begin{cases}
                     K-1, & M = 2K
                    \\
                    K, & M=2K+1,
                 \end{cases}
  \\
  \rho_j, &\quad j= 0, 1, \dots, K,
  \qquad
  \sigma_j,  \quad j = 0, 1, \dots, K-1
\end{align*}
The number of $\rho_j$ and $ \sigma_j$ is common for  even and odd values of $M$. 

The condition $\tilde{a}_{-}\ket{\chi_{00}}=0$ gives the relations:
\begin{align}
  2(M-2j)\rho_{j} - (2j+1)\sigma_{j} = 0,& \quad j = 0, 1, \cdots, K-1 \label{100-1}\\
  j\rho_{j} + (r+M+2j-1)\sigma_{j-1} = 0,& \quad j=1, 2, \cdots, K \label{100-2}
\end{align}
and the additional relation for $M=2K+1:$
\begin{equation}
   \rho_K = 0.
\end{equation}
The condition  $a_{-}\ket{\chi_{00}}=0$ gives the relations:
 \begin{align}
 2M\rho_{0} + \sigma_{0} = 0& \label{100-4}\\
 2(M-2j)\rho_{j} + (2j+1)\sigma_{j} + 8(M-2j)\sigma_{j-1} = 0&,\quad j = 1,2, \cdots, K-1 \label{100-5}\\
 -j\rho_{j} + (r+M-2j-1)\sigma_{j-1} = 0&, \quad j = 1, 2, \cdots, K \label{100-6}
 \end{align}
and the additional relation for $M=2K+1:$
\begin{equation}
 \rho_{K} + 4\sigma_{K-1} = 0 \label{100-7}
\end{equation}

We are able to suppose that  $ M \neq 0 $ as $M=0$ corresponds to the lowest weight space.
Then the direct computation leads to $ \rho_j = \sigma_j = 0 $ for all $j$.
Therefore, we have shown that $\ket{\chi_{00}} = 0.$ 

Next we consider $\ket{\chi_{11}}$. 
The condition $\tilde{a}_{-}\ket{\chi_{11}}=0$ leads the following relations:
\begin{align}
 (M-2j+1)\gamma_{j-1} - j\delta_{j} &= 0,\quad j = 1, 2, \cdots, 
\left\{
\begin{array}{ll}
K-1, & M=2K\\
K, & M=2K+1
\end{array}    
\right. \label{111-1}\\
 (2j+1)\gamma_{j} + 2(r+M+2j)\delta_{j} &= 0,\quad j = 0, 1, \cdots, 
\left\{
\begin{array}{ll}
K-1, & M=2K\\
K, & M=2K+1
\end{array}    
\right. \label{111-2}
\end{align}
and the additional relation for $M=2K$
\begin{align}
\gamma_{K-1} = 0. \label{111-3}
\end{align}
From the condition  $a_{-}\ket{\chi_{11}} = 0$ we have 
\begin{align}
 (M-2j+1)(\gamma_{j-1}+4\delta_{j-1}) + j\delta_{j} &= 0, \quad j = 1,2, \cdots, 
\left\{
\begin{array}{ll}
K-1, & M=2K\\
K, & M=2K+1
\end{array}    
\right. \label{111-4}\\
 (2j+1)\gamma_{j} - 2(r+M-2j-2)\delta_{j} &= 0,\quad j = 0, 1, \cdots, 
\left\{
\begin{array}{ll}
K-1, & M=2K\\
K, & M=2K+1
\end{array}    
\right. \label{111-5}
\end{align}
and for $M=2K$
 \begin{equation}
 \gamma_{K-1} + 4\delta_{K-1} = 0. \label{111-6}
 \end{equation}

For $M=2K,$ it is immediate to see that $ \ket{\chi_{11}} = 0 .$
While, for $M=2K+1$ we have
\[
   \delta_{j} = (-4)^j \dbinom{K}{j}\delta_{0}
\]
and two differrent expressions of $\gamma_j$ from \eqref{111-2} and \eqref{111-5} which are given by
\[
  \gamma_j = -\frac{2}{2j+1} (r+2K+2j+1) \delta_j
\]
and
\[
    \gamma_j = \frac{2}{2j+1} (r+2K-2j-1) \delta_j.
\]
These two formulae coincide if and only if $ r+2K = 0 $ and under this constraint on $r $ and $K$ we have 
\[
  \gamma_j = \frac{1}{2}(-4)^{j+1} \dbinom{K}{j}\delta_{0}.
\]
Therefore, there exists a singular vector $\ket{\chi_{11}} $ in the weight space of level $ 2(2K+1)$ if $ r+2K =0.$ 
Replacing $K$ with $M$ in the above obtained $ \gamma_j, \delta_j,$ 
we get an explicit formula of $\ket{\chi_{11}}$ given by \eqref{111}.   

\medskip
\noindent
(ii) $ M(r,\lambda)$: the $\Z22$-degree of the basis \eqref{BasisMrl} is 
$ (\alpha+m+\beta, k+m+\beta)$. 
The search of singular vectors is carried out in a way similar to that of $M(r).$ 
A main difference from $M(r)$ is that the dimension of each weight space is doubled. 

\medskip
\noindent
(ii-a) Consider a weight space of odd level $2M+1$. 
Then each vector has $\Z22$-degree either $(0,1)$ or $(1,0).$ 
The degree $(0,1)$ and $(1,0)$ basis are given as follows:
\begin{align*} 
(0,1) : &\quad \ket{0,2(M-j)+1,j;\ol{j}}, \quad \ket{1,2(M-j),j;\ol{j+1}} \\[3pt]
(1,0) : &\quad \ket{0,2(M-j)+1,j;\ol{j+1}}, \quad \ket{1,2(M-j),j;\ol{j}}
\end{align*}
where $0 \leq j \leq M$. 
The degree $(1,0)$ basis is obtained from the degree $(0,1)$ basis by exchanging $\ol{j} $ and $\ol{j+1}.$  
This means that a search of degree $(1,0)$ singular vectors is essentially same as degree $ (0,1)$ case provided that $\ol{j}$ and $\ol{j+1}$ are exchanged.  
Therefore, it is enough to carry out a search of degree $(0,1)$ singular vectors.  
The singular vector $\ket{\chi_{01}}$ of degree (0,1) has the form of
\[
\ket{\chi_{01}} = \sum_{j=0}^{M}\big( \mu_{j} \ket{0,2(M-j)+1,j; \ol{j}} + \nu_{j}\ket{1,2(M-j),j; \ol{j+1}} \big).
\]
From the condition $\tilde{a}_{-}\ket{\chi_{01}}=0$ we obtain the relations for the coefficients $\mu_j$ and  $\nu_j$:
\begin{align}
  \lambda^{\overline{M}} \mu_{M} - ( r+2M ) \nu_{M} &= 0,  \label{201-3}
  \\
   (j+1) \mu_{j+1} + 2 \lambda^{\overline{j}} \mu_{j} - 2 ( r + M + j ) \nu_{j} &= 0, \label{201-1}\\
 2 ( M-j )\mu_{j} + (j+1) \nu_{j+1} &= 0  \label{201-2}
\end{align}
where $j$ runs from $0$ to $M-1.$ 
From the condition $a_{-}\ket{\chi_{01}}=0$ we have
\begin{align}
 2 ( r+M ) \mu_{0} - \nu_{1} -2\lambda \nu_{0} &= 0, \quad \label{201-4}\\
 2 ( r+M-j )\mu_{j} - (j+1) \nu_{j+1} - 8( M - j +1) \nu_{j-1} -2 \lambda^{\overline{j+1}}  \nu_{j} &= 0 , \quad  j = 1,2, \cdots, M-1\label{201-6}\\
 r\mu_{M} - 4 \nu_{M-1} - \lambda^{\overline{M+1}} \nu_{M} &= 0, \label{201-7} \\
 (j+1) \mu_{j+1} + 2 ( M-j ) \nu_{j} &= 0 , \quad  j = 0,1,\cdots, M-1  \label{201-5}
\end{align}
where $\lambda \neq 0.$

Getting rid of $ \nu_j $ from \eqref{201-1} by \eqref{201-5}  one has
\begin{equation}
   (j+1) (r+2M)\mu_{j+1} + 2\lambda^{\ol{j}}(M-j) \mu_j = 0, \quad j= 0, 1, \dots, M-1
   \label{mu-rec-01}
\end{equation}
For $r+2M = 0$, one may see $ \mu_j= \nu_j = 0$ so that $\ket{\chi_{01}}=0$.
For $r+2M \neq 0$, the explicit formulae of $\mu_j$ and $\nu_j$ are given as
\begin{equation}
\mu_{j} = (-2)^j \dfrac{\lambda^{\left\lfloor \frac{j}{2} \right\rfloor}}{(r+2M)^j}  \dbinom{M}{j}  \mu_{0}. \label{2mu}
\end{equation}
and
\begin{equation}
 \nu_{j} = (-2)^j \dfrac{ \lambda^{ \left\lfloor \frac{j+1}{2} \right\rfloor }}{( r + 2M )^{j+1}} \dbinom{M}{j} \mu_{0}, \quad  j = 0,1,\cdots, M. \label{2nu}
\end{equation}
To solve \eqref{201-7}, we find that the constraint $(r+2M)^2 = \lambda$ is necessary. 
It is then straightforward to verify that \eqref{2mu} and \eqref{2nu} with this constraint solve all the relations. 
Thus we have non-vanishing $\mu_j, \nu_j$ given by
\[
  \mu_{j} = (-2)^j \frac{(r+2M)^{\ol{j+1}}}{r+2M} \dbinom{M}{j} \mu_{0},
  \qquad
  \nu_{j} = (-2)^j \frac{(r+2M)^{\ol{j}}}{r+2M} \dbinom{M}{j} \mu_{0}.
\]
Therefore, we have shown the existence of $\ket{\chi_{01}} $ given by \eqref{201} for $ (r+2M)^2 = \lambda.$ 

As we already mentioned, the search of degree $(1,0)$ singular vectors  is same as the case of degree $(0,1)$ provided that $ \ol{j} $ and $ \ol{j+1}$ are exchanged. 
Thus we conclude there exists $\ket{\chi_{10}}$ given by \eqref{210} if $ (r+2M)^2 = \lambda.$ 

\medskip
\noindent
(ii-b) Finally, we consider a weight space of even level $2M.$ 
The basis is given by
\begin{align*}
(0,0) : &\quad \ket{0,2(M-j),j;\ol{j}}, \quad j = 0,1,\cdots, M\\
&\quad \ket{1,2(M-j)-1,j;\ol{j+1}}, \quad j=0,1\cdots, M-1\\[5pt]
(1,1) :& \quad \qquad \ol{j} \quad \longleftrightarrow \quad \ol{j+1} 
\end{align*}
where the basis of degree $(1,1)$ is obtained from that of degree $(0,0)$ exchanging $\ol{j}$ and  $\ol{j+1}$.
Therefore, as before, it is sufficient to look for only the degree $(0,0)$ singular vector $\ket{\chi_{00}}$. 
$\ket{\chi_{00}}$ is written as follows:
\[
\ket{\chi_{00}} = \sum_{j=0}^{M} \rho_{j} \ket{0,2(M-j),j;\ol{j}} + \sum_{j=0}^{M-1} \sigma_{j} \ket{1,2(M-j)-1,j;\ol{j+1}}.
\]
One obtain the following relations from the condition $\tilde{a}_{-}\ket{\chi_{00}}=0:$
\begin{align}
 (j+1) \rho_{j+1} + 2(r+M+j)\sigma_{j} &= 0, \quad j = 0,1,\cdots, M-1 \label{200-1}\\
 2(M-j)\rho_{j} -(j+1) \sigma_{j+1} - 2\lambda^{\ol{j+1}}\sigma_{j} &= 0, \quad j = 0,1,\cdots, M-2 \label{200-2}\\
 \rho_{M-1} - \lambda^{\ol{M}}\sigma_{M-1} &= 0 \label{200-3}
\end{align}
and also from $a_{-}\ket{\chi_{00}}=0:$
\begin{align}
 2M\rho_{0} - \sigma_{1} + 2\lambda\sigma_{0} &= 0 \label{200-4}\\
 2(M-j)\rho_{j} - (j+1) \sigma_{j+1} + 8(M-j)\sigma_{j-1} + 2\lambda^{\ol{j+1}}\sigma_{j} &= 0, \quad 
 j = 1,2,\cdots, M-2 \label{200-6}\\
\rho_{M-1} + 4\sigma_{M-2} + \lambda^{\ol{M}}\sigma_{M-1} &= 0 \label{200-7} 
\\
 (j+1)\rho_{j+1} - 2(r+M-j-2)\sigma_{j} &= 0,\quad j = 0,1,\cdots, M-1 \label{200-5}
\end{align}
It is not difficult to see that these recurrence relations lead to $\rho_j = \sigma_j = 0$ for all $j.$ Therefore we get $\ket{\chi_{00}}=0$.
$ \ket{\chi_{11}} = 0 $ can be shown in the same way provided that $\ol{j}$ and $ \ol{j+1}$ are exchanged. 
\end{proof}

%

\section{Irreducible modules} \label{SEC:IrMod}
\setcounter{equation}{0}

Our final task is to list up all irreducible modules. 
It is immediate from Proposition \ref{prop:mis} and Theorem \ref{thm:sv} that if $ r+2M \neq 0 $ then the Verma module $M(r)$ is irreducible and if $ (r+2M)^2 \neq \lambda $ then 
the Verma module $M(r,\lambda)$ is irreducible. 
This proves the cases (i) and (iii) of Theorem \ref{thm:ILWM}. 
In order to verify the cases (ii) and (iv), we need to specify the maximal invariant submodule $W.$ 
Let $ \omega $ be the subspace of $ M(r) $ or $ M(r,\lambda)$ spanned by the singular vectors $ \ket{\chi_{01}}$ and 
$ \ket{\chi_{10}}$ of Theorem \ref{thm:sv}. 
Then $ W = U(\n^+) \otimes \omega $ for both $ M(r) $ and $ M(r,\lambda).$ 
This is obvious for $ M(r,\lambda)$ due to Proposition \ref{prop:mis}, but highly non-trivial for $ M(r).$ 
First, we establish this fact for $M(r)$ and then investigate  the dimension  of irreducible modules. 
This fact is a consequence of the next Proposition.

\begin{prop}\label{prop:Mrsv}
The singular vector $ \ket{\chi_{11}} \in M(r)$  given in \eqref{111} belongs to $  U(\n^+) \otimes \omega. $ 
\end{prop} 

This proposition indicates that the invariant submodule $ U(\n^+) \otimes \ket{\chi_{11}} $ is a subspce of  $U(\n^+) \otimes \omega. $ 
Thus the maximal invariant submodule of $M(r)$ is given by $ W = U(\n^+) \otimes \omega. $

\begin{proof}
 Suppose that $\ket{\chi_{01}}, \ket{\chi_{10}}$ exist in the level $2M+1$ weight space $M_{2M+1}(r)$, 
then $\ket{\chi_{11}}$ is in the level $2(2M+1)$ weight space $M_{2(2M+1)}(r) $ (Theorm \ref{thm:sv}). 
We shall show that $ \ket{\chi_{11}}$ is a linear combination of the vectors in 
the level $2(2M+1) $ weight space $ (U(\n^+) \otimes \omega)_{2(2M+1)}$. 
The degree $(1,1)$ vectors in $ (U(\n^+) \otimes \omega)_{2(2M+1)}$ are given as follows:
\begin{equation}
  \tilde{a}_{+}a_{+}^{2(M-2k)}\tilde{L}_{+}^{2k} \ket{\chi_{01}},
  \quad
  \quad a_{+}^{2(M-2k)+1}\tilde{L}_{+}^{2k}\ket{\chi_{10}},
  \quad
  k = 0, 1, \dots, \left\lfloor \frac{M}{2} \right\rfloor 
  \label{L11vec1}
\end{equation}
and 
\begin{equation}
  a_{+}^{2(M-2k)-1}\tilde{L}_{+}^{2k+1}\ket{\chi_{01}},
  \quad
  \tilde{a}_{+}a_{+}^{2(M-2k-1)}\tilde{L}_{+}^{2k+1} \ket{\chi_{10}},
  \quad
  k = 0, 1, \dots, \left\lfloor \frac{M-1}{2} \right\rfloor
  \label{L11vec2}
\end{equation}
Using the explicit formula of $ \ket{\chi_{01}}, \ket{\chi_{10}}$, 
one may verify by straightforward computation that the sum of the vectors in \eqref{L11vec1} and \eqref{L11vec2} are given as follows:
\begin{align*}
  & \tilde{a}_{+}a_{+}^{2(M-2k)}\tilde{L}_{+}^{2k} \ket{\chi_{01}}
  +
   a_{+}^{2(M-2k)+1}\tilde{L}_{+}^{2k}\ket{\chi_{10}}
  \nn \\
  & \quad
   =
   \sum_{j=k}^{\lfloor \frac{M}{2} \rfloor +k} 
   2^{2j+1-2k} \dbinom{M}{2j-2k} 
   \big( -2\ket{0,4(M-j),2j+1} + \ket{1,4(M-j)+1,2j} \big),
  \\[5pt]
  &  a_{+}^{2(M-2k)-1}\tilde{L}_{+}^{2k+1}\ket{\chi_{01}}
  +
  \tilde{a}_{+}a_{+}^{2(M-2k-1)}\tilde{L}_{+}^{2k+1} \ket{\chi_{10}}
  \nn \\
  & \quad 
  = 
  \sum_{j=k+1}^{\lfloor \frac{M+1}{2}\rfloor +k}
  2^{2j-2k} \dbinom{M}{2j-2k-1} \big( -2\ket{0,4(M-j),2j+1} + \ket{1,4(M-j)+1,2j} \big).
\end{align*}
A linear combination of these vectors
\begin{align*}
  &\sum_{k=0}^{\left\lfloor \frac{M}{2} \right\rfloor} c_{2k} 
   ( 
     \tilde{a}_{+}a_{+}^{2(M-2k)}\tilde{L}_{+}^{2k} \ket{\chi_{01}}
     +
     a_{+}^{2(M-2k)+1}\tilde{L}_{+}^{2k}\ket{\chi_{10}}
   )
   \\
   & \qquad\qquad
  + \sum_{k=0}^{\left\lfloor \frac{M-1}{2} \right\rfloor} c_{2k+1}
  (
    a_{+}^{2(M-2k)-1}\tilde{L}_{+}^{2k+1}\ket{\chi_{01}}
  +
  \tilde{a}_{+}a_{+}^{2(M-2k-1)}\tilde{L}_{+}^{2k+1} \ket{\chi_{10}}
  )
\end{align*}
coincides with $ \ket{\chi_{11}}$ given in \eqref{111} if the coefficients $ c_p$ satisfies the relations:
\begin{align*}
 \sum_{p=0}^{2j} 2^{2j+1-p} \dbinom{M}{2j-p} c_p
 &= (-4)^j \dbinom{M}{j},
 \quad j = 0, 1, \dots, \left\lfloor \frac{M}{2} \right\rfloor
 \\
 \sum_{p=2j-M}^{M} 2^{2j+1-p} \dbinom{M}{2j-p} c_{p}
 &= (-4)^j \dbinom{M}{j}, 
 \quad j=  \left\lfloor \frac{M}{2} \right\rfloor+1, \cdots, M
\end{align*}
This is a system of linear algebraic equations for $ c_p. $ 
We give $M=4$ and $M=5$ coefficient matrix as an example. 
$  M = 4: $
\[
  \begin{pmatrix}
    2\dbinom{4}{0} & 0 & 0 & 0 & 0 
    \\[12pt]
    2^3 \dbinom{4}{2} & 2^2\dbinom{4}{1} & 2 \dbinom{4}{0} & 0 & 0 
    \\[12pt]
    2^5 \dbinom{4}{4} & 2^4 \dbinom{4}{3} & 2^3 \dbinom{4}{2} & 2^2 \dbinom{4}{1} & 2 \dbinom{4}{0}
    \\[12pt]
    0 & 0 & 2^5 \dbinom{4}{4} & 2^4\dbinom{4}{3} & 2^3 \dbinom{4}{2}
    \\[12pt]
    0 & 0 & 0 & 0 & 2^5 \dbinom{4}{4}
  \end{pmatrix}.
\] 
$  M = 5: $
\[
  \begin{pmatrix}
    2\dbinom{5}{0} & 0 & 0 & 0 & 0 & 0
    \\[12pt]
    2^3 \dbinom{5}{2} & 2^2\dbinom{5}{1} & 2 \dbinom{5}{0} & 0 & 0 & 0 
    \\[12pt]
    2^5 \dbinom{5}{4} & 2^4 \dbinom{5}{3} & 2^3 \dbinom{5}{2} & 2^2 \dbinom{5}{1} & 2 \dbinom{5}{0} & 0
    \\[12pt]
    0 &  2^6 \dbinom{5}{5} & 2^5 \dbinom{5}{4} & 2^4\dbinom{5}{3} & 2^3 \dbinom{5}{2} 
    & 2^2 \dbinom{5}{1}
    \\[12pt]
    0 & 0 & 0 & 2^6 \dbinom{5}{5} & 2^5 \dbinom{5}{4} & 2^4\dbinom{5}{3} 
    \\[12pt]
    0 & 0 & 0 & 0 & 0 & 2^6 \dbinom{5}{5}
  \end{pmatrix}.
\] 
These matrices are not singular because one may make them triangular  with no non-zero diagonal entries. 
Thus $c_p$'s are determined uniquely and Proposition \ref{prop:Mrsv} is proved for $M=4, 5.$ 

The general form of the coefficient matrix, denoted by $A,$ is as follows: 
$A$ is a $ (M+1) \times (M+1)$ matrix and its entry is $0$ or one of the integers in the sequence:
\begin{equation}
 2^{M+1} \dbinom{M}{M}, \quad 
 2^M \dbinom{M}{M-1}, \quad 2^{M-1} \dbinom{M}{M-2}, \quad \dots, \quad 
 2\dbinom{M}{0}.
 \label{sequence}
\end{equation}
The non-vanishing entry of the first row of $A$ is the right most term of \eqref{sequence}: $ A_{11} = 2\dbinom{M}{0}. $ 
The non-vanishing entries of the second row consist of the three terms from the right of \eqref{sequence}:
\[
  A_{21} = 2^3 \dbinom{M}{2}, \quad A_{22}=  2^2 \dbinom{M}{1},
  \quad
  A_{23} = 2\dbinom{M}{0}
\]
In the third row two more terms are taken from \eqref{sequence}:
\[
  A_{31} = 2^5 \dbinom{M}{4}, \quad 
  A_{32} = 2^4 \dbinom{M}{3}, \quad 
  A_{33} = 2^3 \dbinom{M}{2}, \quad 
  A_{34} = 2^2 \dbinom{M}{1}, \quad 
  A_{35} = 2 \dbinom{M}{0}.
\]
In this way, as row goes further by one, two more terms are taken from \eqref{sequence}. 
If $M$ is even, the $ (\frac{M}{2}+1)$th row is the sequence \eqref{sequence} itself. 
In the $(\frac{M}{2}+2)$th row the sequence moves right by two columns and $0$ is put in the empty slot of the row. 
Repeating this process, the $(M+1)$th row of $A$ has only one  non-vanishing entry 
$ A_{M+1,M+1} = 2^{M+1} \dbinom{M}{M}.$ 
If $M$ is odd, there is a slight difference in the middle of this process. 
That is, the $\frac{1}{2}(M+1)$th row reads
\[
 2^M \dbinom{M}{M-1}, \quad 2^{M-1} \dbinom{M}{M-2}, \quad \dots, \quad 
 2\dbinom{M}{0}, \quad 0
\]
and this sequence moves by two columns which gives the $\frac{1}{2}(M+3)$th row as follows:
\[
 0, \quad 2^{M+1} \dbinom{M}{M}, \quad 
 2^M \dbinom{M}{M-1}, \quad 2^{M-1} \dbinom{M}{M-2}, \quad \dots, \quad 
 2^2\dbinom{M}{1}
\]
Thus, because of the same reason as the $M=4$ and $M=5$ examples the coefficient matrix is not singular so that $c_p$'s are determined uniquely. 
This completes the proof of Proposition \ref{prop:Mrsv}. 
\end{proof}

In this way, we could identify the maximal invariant submodules of $ M(r) $ and $ M(r,\lambda).$ 
We now turn to count the dimension of the quotient modules. 
This can be done by counting the dimension of each weight space of $W.$ 
Since $W $ has the common structure for $ M(r) $ and $ M(r,\lambda),$ 
that is, $W$ is constructed by repeated application of the elements of $\n^+$ on the two dimensional lowest weight space $\omega$ spanned by $ \ket{\chi_{01}} $ and $\ket{\chi_{10}},$ 
$ \dim W_N$ is also common for $ M(r) $ and $ M(r,\lambda).$ 
Recalling that the singular vectors $ \ket{\chi_{01}}, \ket{\chi_{10}}$ exist in the weight space $W_{2M+1}$, one may prove the following:
\begin{prop}\label{prop:2}
 For a non-negative integer $q$, $ \dim W_{2M+1+q} = 2(q+1).$ 
 Let $\ket{\chi}$ be $\ket{\chi_{01}}$ or $\ket{\chi_{10}},$ 
 then the basis of $ W_{2M+1+q} $ is given by
 \begin{equation}
    (\tilde{a}_{+}a_{+})^{j}\, a_{+}^{q-2j}\ket{\chi}, \qquad
    (a_{+}\tilde{a}_{+})^{j}\, a_{+}^{q-2j}\ket{\chi} 
    \label{Wbasis}
 \end{equation}
 where $j$ runs from $0$ to $ \frac{q}{2}$ for $q$ even and from $0$ to $ \frac{q-1}{2}$ for $q$ odd. 
 In addition, there exist one more vector for $q$ odd which is given by $ (\tilde{a}_{+}a_{+})^{\frac{q-1}{2}}\, \tilde{a}_{+}\ket{\chi}.$
\end{prop} 
\begin{proof}
 We prove the proposition by induction. 
 The case of $ q = 0$ corresponds to the lowest weight space $\omega$ itself so that the proposition is true. 
The $q=1$ weight space $ W_{2M+2}$ is constructed by application of $ \tilde{a}_+, a_+$ just once on $ \omega. $ 
The four vectors  $\tilde{a}_{+}\ket{\chi}, a_{+}\ket{\chi} $  are obviously linearly independent so that the proposition holds true for this case, too.

We now suppose that the proposition is true for any integers not greater than $q.$ 
If $q$ is even, the weight space $ W_{2M+1+(q+1)}$ is spanned by the vectors which are obtained by the action of $ \tilde{a}_+, a_+$ just once on the vectors \eqref{Wbasis}
\begin{alignat}{2}
  & \tilde{a}_+ (\tilde{a}_{+}a_{+})^{j}\, a_{+}^{q-2j}\ket{\chi}, 
  & \quad &
   a_+ (\tilde{a}_{+}a_{+})^{j}\, a_{+}^{q-2j}\ket{\chi}, 
  \nn \\[3pt]
  &\tilde{a}_+ (a_{+}\tilde{a}_{+})^{j}\, a_{+}^{q-2j}\ket{\chi}, 
  &  & 
   a_+ (a_{+}\tilde{a}_{+})^{j}\, a_{+}^{q-2j}\ket{\chi}.   
   \label{Wqp1} 
\end{alignat}
One may have more vectors by the action of $ L_+, \tilde{L}_+ $ on $ W_{2M+1+(q-1)}$ 
since they raise the weight by two. 
However, vectors obtained in such a way are not linearly independent from \eqref{Wqp1} 
because of the defining relations \eqref{Z22osp12}: $L_+ \sim a_+^2, \tilde{L}_+ \sim [a_+, \tilde{a}_+]$.
%
%
We here use the following identities (we used \eqref{Z22osp12} for the first identity):
\begin{alignat*}{2}
 \tilde{a}_{+}(\tilde{a}_{+}a_{+})^j &= -(a_{+}\tilde{a}_{+})^{j-1} \,a_{+}^3,
 &\quad 
 \tilde{a}_{+}(a_{+}\tilde{a}_{+})^j &= (\tilde{a}_{+}a_{+})^j \,\tilde{a}_{+}
 \nn \\[3pt]
 a_{+}(\tilde{a}_{+}a_{+})^j &= (a_{+}\tilde{a}_{+})^{j} \,a_{+},
 &
  a_{+}(a_{+}\tilde{a}_{+})^j &= (\tilde{a}_{+}a_{+})^{j} \,a_{+}
\end{alignat*} 
With these identities one may extract the following linearly independent vectors from the vectors given in \eqref{Wqp1}:
\[
  (a_+ \tilde{a}_+)^j \, a_+^{q+1-2j} \ket{\chi}, \qquad
  (\tilde{a}_+ a_+)^j \, a_+^{q+1-2j} \ket{\chi}, \quad
  (\tilde{a}_+ a_+)^{q/2} \, \tilde{a}_+  \ket{\chi}
\]
where $ j$ runs from $1$ to $q/2$ and the last vector is obtained from $ \tilde{a}_+ (a_+ \tilde{a}_+)^{q/2} \ket{\chi}. $ 
Thus if $q$ is even, the proposition also holds true for $q+1.$ 

If $q$ is odd, one may repeat the same argument to prove that the proposition is true for $q+1.$ So we omit to present the detail. 
This completes the proof of Proposition \ref{prop:2}. 
\end{proof}

Proposition \ref{prop:2} enables us to count the dimension of the quotient modules. 
First, let us recall that the dimension of the weight space of $ M(r) $ and $ M(r,\lambda):$
\[
   \dim M_N(r) = N+1, \qquad \dim M_N(r,\lambda) = 2(N+1).
\]
Set $N=2M+1$ then we see from Proposition \ref{prop:2}
\[
  \dim (M(r)/W)_{N+q} = 2M-q, \qquad \dim (M(r,\lambda)/W)_{N+q} = 4M+2.
\] 
It follows that $ \dim (M(r)/W) = (2M+1)^2 $ and $ \dim (M(r,\lambda)/W) = \infty.$ 

Having this result on dimension the proof of Theorem \ref{thm:ILWM} has been completed.

%
%
%
\section{Concluding remarks}

We presented a classification of the irreducible lowest weight modules over the $\Z22$-$osp(1|2).$ 
The major difference from the $osp(1|2)$ Lie superalgebra is the existence of infinite dimensional modules with the degenerate lowest weight. 
The $\Z22$-$osp(1|2) $ also has finite and infinite dimensional modules constructed on a non-degenerate lowest weight vector. 
These are the $\Z22$-graded counterparts of the irreducible modules of $osp(1|2).$ 
Thus the structure of the irreducible modules over the $\Z22$-$osp(1|2) $ is richer than those of $osp(1|2).$  

For this classification we employ the method similar to the strange Lie superalgebra $Q(n)$ 
and the method works quite well. 
This suggests that it is possible to build a general representation theory of $\Z22$-graded simple Lie superalgebras via a procedure similar to simple Lie superalgebras.  
Another important issue which is not dealt with in this work is a classification of unitary representations. 
This class of representations is also important to applications for quantum physics. 
These will be  future works. 

We make a connection of the present results with the $\Z22$-graded version of superconformal mechanics discussed in \cite{AAD}. 
The Hamiltonian of the $\Z22$-$osp(1|2)$ superconformal mechanics corresponds to  $R$ of the present paper and the ground state energy, which is our $r,$ is given as $ \frac{1}{2}(1-2\beta) $ where $\beta$ is the coupling constant of the model. 
The ground state and all the excited states are doubly degenerate. 
This corresponds to the case (iv) of Theorem \ref{thm:ILWM} with $M=0$ so that 
$ \lambda = r^2 = \frac{1}{4} (1-2\beta)^2.$ 
The model of \cite{AAD} is a quantum mechanics of single particle. 
The irreducible module with higher values of $M$ will be realized in  multipartite quantum systems since excitation of different particle creates different excited states. 
A multiparticle quantum system with $\Z22$-graded supersymmetry is discussed in \cite{Topp}. 

It is widely known that the finite dimensional modules of $osp(1|2)$ are realized in various physical problems. 
It will be  interesting to investigate physical realization of the $\Z22$-$osp(1|2)$ counterpart ((ii) of Theorem \ref{thm:ILWM}). 
Finite dimensional irreducible modules of both $osp(1|2)$ and $\Z22$-$osp(1|2)$ are of odd dimension.  
However, it is known that the quantum algebra $ U_q[osp(1|2n)]$ has even dimensional irreducible module \cite{Zou}. 
The even dimensional representations of $U_q[osp(1|2)]$ have a connection with $q$-Hahn and little $q$-Jacobi polynomials and they are used to construct noncommutative spaces \cite{ACNS}. 
In this respect, one may think that investigations of quantum analogue of $\Z22$-$osp(1|2)$ will provide a new perspective to orthogonal polynomials and noncommutative geometry 
(see \cite{BruDupq} non-commutative $\Z22$-graded $q$-plane).

%
%
 \appendices
 \appendixpage
 \section{$\Z22$-graded Lie superalgebra}
 \setcounter{equation}{0}

Here we give the definition of a $\Z22$-graded color superalgebra \cite{rw1,rw2}. 
Let $ \g $ be a vector space  and $ \bm{a} = (a_1, a_2)$ an element of $\Z22$. 
Suppose that $ \g $ is a direct sum of graded components:
\begin{equation}
   \g = \bigoplus_{\bm{a}} \g_{\bm{a}} = \g_{(0,0)} \oplus \g_{(0,1)} \oplus \g_{(1,0)} \oplus \g_{(1,1)}.
\end{equation}
In what follows, we denote homogeneous elements of $ \g_{\bm{a}} $ as $ X_{\bm{a}}, Y_{\bm{a}},
Z_{\bm{a}}$. 
If $\g$ admits a bilinear operation (the general Lie bracket), denoted by $ \llbracket \cdot, \cdot \rrbracket, $ 
satisfying the identities
\begin{align}
  & \llbracket X_{\bm{a}}, Y_{\bm{b}} \rrbracket \in \g_{\bm{a}+\bm{b}}
  \\[3pt]
  & \llbracket X_{\bm{a}}, Y_{\bm{b}} \rrbracket = -(-1)^{\bm{a}\cdot \bm{b}} \llbracket Y_{\bm{b}}, X_{\bm{a}} \rrbracket,
  \\[3pt]
  & (-1)^{\bm{a}\cdot\bm{c}} \llbracket X_{\bm{a}}, \llbracket Y_{\bm{b}}, Z_{\bm{c}} \rrbracket \rrbracket
    + (-1)^{\bm{b}\cdot\bm{a}} \llbracket Y_{\bm{b}}, \llbracket Z_{\bm{c}}, X_{\bm{a}} \rrbracket \rrbracket
    + (-1)^{\bm{c}\cdot\bm{b}} \llbracket Z_{\bm{c}}, \llbracket X_{\bm{a}}, Y_{\bm{b}} \rrbracket
\rrbracket =0, 
    \label{gradedJ}
\end{align}
where
\begin{equation}
  \bm{a} + \bm{b} = (a_1+b_1, a_2+b_2) \in \Z22, \qquad \bm{a}\cdot \bm{b} = a_1 b_1 + a_2 b_2\in {\mathbb Z}_2,
\end{equation}
then $\g$ is referred to as a $\Z22$-graded color superalgebra. 

The enveloping algebra of $\g$ is the $\Z22$-graded unital associative algebra with relations
\begin{equation}
  X_{\bm{a}} Y_{\bm{b}} - (-1)^{\bm{a}\cdot \bm{b}} Y_{\bm{b}} X_{\bm{a}} = \llbracket X_{\bm{a}}, Y_{\bm{b}} \rrbracket.
  \label{gradedcom}
\end{equation}
One key observation is that for homogeneous elements, their general Lie bracket will coincide with either a commutator $(\bm{a}\cdot \bm{b}=0)$ or anticommutator $(\bm{a}\cdot \bm{b}=1)$ .  
In the main body of this paper, we also used the notation 
$ [X_{\bm{a}}, Y_{\bm{b}}] $ (in case $ \bm{a}\cdot \bm{b}=0$) and 
$ \{ X_{\bm{a}}, Y_{\bm{b}}  \} $ (in case $ \bm{a}\cdot \bm{b}=1$) 
for the general Lie bracket in order to emphasize that given elements commute or anticommute.
%
%
It should be noted that  $ \g_{(0,0)} \oplus \g_{(0,1)} $ and $ \g_{(0,0)} \oplus \g_{(1,0)} $ are
subalgebras of $\g$ (with $\Z22$-grading).
We remark that this is a natural generalization of Lie superalgebra which is defined with a $\mathbb{Z}_2$-graded structure:
\begin{equation}
  \g =  \bigoplus_{\bm{a}} \g_{\bm{a}} = \g_{(0)} \oplus \g_{(1)},
\end{equation}
instead with 
\begin{equation}
  \bm{a} + \bm{b} = (a+b)\in {\mathbb Z}_2, \qquad \bm{a} \cdot \bm{b} = ab\in {\mathbb Z}_2.
\end{equation} 
 
%
%
\section*{Data Availability}

Data sharing is not applicable to this article as no new data were created or analyzed in this study.
 
%

\end{document}